\documentclass[twocolumn,reprint,superscriptaddress]{revtex4-2}
\usepackage{amsmath}
\usepackage{amssymb}
\DeclareMathOperator*{\Max}{Max}
\usepackage{eufrak}
\usepackage[colorlinks=true,linkcolor=blue,citecolor=blue, urlcolor=blue]{hyperref}  
\usepackage{graphicx}
\usepackage{lipsum}
\usepackage{xcolor}
\usepackage{blindtext}
\usepackage{amsthm}
\newtheorem{theorem}{Theorem}[section]

\begin{document}

	\title{Absolute fully entangled fraction from spectrum}

	\author{Tapaswini Patro}
	\email[]{p20190037@hyderabad.bits-pilani.ac.in}
	
	\affiliation{Department of Mathematics, Birla Institute of Technology and Science-Pilani, Hyderabad Campus,Telangana,India}
	
	\author{Kaushiki Mukherjee}
	\email[]{kaushiki.wbes@gmail.com}
	
	\affiliation{Department of Mathematics, Government Girls' General Degree College, Ekbalpore, Kolkata, India}
	
	\author{Mohd Asad Siddiqui}
	\email[]{asad@ctp-jamia.res.in}
	
	\affiliation{Institute of Fundamental and Frontier Sciences, University of Electronic Science and Technology of China, Chengdu 610051, China}
	
	\author{Indranil Chakrabarty}
	\email[]{indranil.chakrabarty@iiit.ac.in}
	
	\affiliation{Center for Quantum Science and Technology and Center for Security, Theory and Algorithmic Research, International Institute of Information Technology-Hyderabad, Gachibowli, Telangana, India}
	
	\author{Nirman Ganguly}
	\email[]{nirmanganguly@hyderabad.bits-pilani.ac.in}
	
	\affiliation{Department of Mathematics, Birla Institute of Technology and Science-Pilani, Hyderabad Campus,Telangana,India}

	\begin{abstract}
		Fully entangled fraction (FEF) is a significant figure of merit for density matrices. In bipartite $ d \otimes d $ quantum systems, the threshold value  FEF $ > 1/d $, carries significant implications for quantum information processing tasks. Like separability, the value of FEF is also related to the choice of global basis of the underlying Hilbert space. A state having its FEF $ \le 1/d $, might give a value $ > 1/d $ in another global basis. A change in the global basis corresponds to a global unitary action on the quantum state. In the present work, we find that there are quantum states whose FEF remains less than $ 1/d $, under the action of any global unitary i.e., any choice of global basis. We invoke the hyperplane separation theorem to demarcate the set from states whose FEF can be increased beyond $ 1/d $ through global unitary action. Consequent to this, we probe the marginals of a pure three party system in qubits. We observe that under some restrictions on the parameters, even if two parties collaborate (through unitary action on their combined system) they will not be able to breach the FEF threshold. The study is further extended to include some classes of mixed three qubit and three qutrit systems. Furthermore, the implications of our work pertaining to $ k- $copy nonlocality and teleportation are also investigated.     
	\end{abstract}

	\maketitle

	\section{Introduction}
	\par Quantum information theory \cite{nie} promises to offer significant advantages over conventional classical procedures. It relies on correlations significantly different from the ones possible in the classical realm. One such class of correlations is exhibited by quantum states which are entangled. Quantum entanglement \cite{ent}, is said to be possessed by states which cannot be written as a convex combination of separable(not entangled) ones. Entanglement is necessary for protocols like teleportation \cite{tele}, entanglement swapping \cite{swap}, dense coding, \cite{dense} and in certain cryptographic schemes \cite{QC1}. Thus entanglement can be considered as a resource and hence studies on entanglement have ranged from its applications as mentioned above to more foundational issues \cite{ent}. Two decades back researchers have witnessed the existence of correlations that go beyond the notion of entanglement \cite{oneq}. Later authors came up with various theoretical measures of such  correlations and named them as discord \cite{disc1, disc2}, dissonance \cite{disso}, dissension \cite{dissen, dissen1}.  
	\par In the entanglement-separability paradigm, one such foundational issue is the absolute separability problem \cite{knill}. The ability of a quantum state to be entangled depends on the choice of the basis of the underlying Hilbert space \cite{Hahn}. A state which is entangled in one basis might be separable in another basis. However, there are states which remain separable under any change of basis within a maximal ball around the maximally mixed state, which was first observed in \cite{purity1}. Such states are termed as absolutely separable \cite{purity2}. It can be seen that, absolutely separable states $ \sigma_{as} $ are specifically those for which $ U \sigma_{as} U^\dagger $ is separable for all possible global unitary operations $ U $ and hence any criteria to detect absolutely separable states will be based on the spectrum of density matrices. Separability of states close to maximally mixed states have been studied \cite{mixed} and found to be absolutely separable for a certain purity threshold \cite{purity1}. Necessary and sufficient conditions based on the spectrum have been found to detect absolute separability for two qubits and qubit-qudit states \cite{Verst,johnston}. However, literature still lacks a complete criterion for any arbitrary dimension \cite{Absolute separability}. 
	\par Since, generation of entanglement from previously separable states is an area of active experimental work \cite{genent}, the set containing absolutely separable states was characterized in \cite{Ganguly14}. The work \cite{Ganguly14} pertained to the detection of non-absolutely separable states, through linear hermitian operators from which entanglement can be generated. This study was further extended through the study on extreme points of the absolute separable class \cite{extreme} and non-linear detection of non-absolutely separable states \cite{patra}. The idea of absolute class of states is not only restricted to separable states but had also been shown in the context of Bell-CHSH local states \cite{Bell-CHSH} and also for states with non negative conditional entropies \cite{spatro, vem,vem1}.
	\par As noted before, entanglement being a necessary condition for quantum information protocols, its mere presence does not guarantee success. We need to look for different manifestations of entangled states. One key feature associated with quantum states is its fully entangled fraction(FEF) \cite{horodecki}. It measures the proximity of a quantum state to maximally entangled states \cite{properties}. For bipartite systems in $ d \otimes d $ dimensions, the threshold value of FEF $ > \frac{1}{d}  $ is a significant benchmark as quantum states with such a figure of merit prove to be useful in teleportation \cite{wernerstate} and entanglement swapping \cite{swap}. Some significant works on FEF can be found in \cite{upper bound,Grondalski,maximum eigenvalue,wernerstate,frobenous}. A prescription of detecting entangled states with FEF $ > \frac{1}{d}  $ was laid in \cite{Ganguly11}. 
	\par However, FEF of a state also gets perturbed with the change in the basis of the Hilbert space. A state with FEF$ \le \frac{1}{d} $ in one basis can have its FEF increased beyond $ \frac{1}{d} $ through global unitary action. In the present work, we introduce the notion of states having absolute fully entangled fraction, specifically it is the class of states whose FEF cannot be increased beyond $ \frac{1}{d} $ even with global unitary action. A mathematical definition of the class is given in section \ref{pre}. The class is shown to be convex and compact. This paves the way for constructing suitable operators (witnesses) to detect states whose FEF can be increased beyond the threshold value $ \frac{1}{d} $. The identification thus buttresses the generation of entangled states efficient for teleportation,through action of suitable unitary gates akin to entanglement generation. Illustrative examples are given with experimental prescription of the witness operators pertaining to both qubit and qutrit systems. A necessary and sufficient criterion based on the spectrum of states is derived to ascertain membership in the class. Separable states have their FEF bounded by $ \frac{1}{d} $. Hence, absolutely separable states cannot have their FEF increased beyond $ \frac{1}{d} $ by global unitary operations. This establishes that the class of absolutely separable states is a subset of the class pertaining to absolute fully entangled fraction. We show that this inclusion is strict. For some particular classes of states, characterization is also made in terms of the Bloch parameters. A numerical study based on purity of states is further carried out to highlight the necessary underpinnings. We note that all the analyses have been done for states living in bipartite systems. 
	\par The act of a global unitary action makes sense when two parties collaborate to do the operation. Consider that a tripartite system is shared between three parties and the situation when two of them collaborate. We show here that under certain restrictions, even if two parties collaborate they will not be able to breach the FEF threshold(FEF $ \le 1/d $) with a joint unitary action. This study is done for all three qubit pure states and some mixed states in three qubits and three qutrits. Our work has implications for teleportation and $ k- $ copy nonlocality as  FEF $ > 1/d $ is a significant benchmark for those tasks. We show how we can identify states which can be made useful for the tasks using a global unitary action.  
	\par Our work is structured in the following way in this paper: In sec.\ref{pre} we revisit the related definitions and establish the notations to be used in the paper. In sec. \ref{charac}, we discuss the criterion for membership in $ \mathfrak{AF} $[set containing states having absolute fully entangled fraction] in a two-qudit system. We characterize $ \mathfrak{AF} $ here and invoke the hyperplane separation theorem for the construction of witness operators. Sec. \ref{decom} discusses relevant illustrations and also the decomposition of the operator in terms of local observables. Characterization of the set in terms of Bloch parameters and bounds on purity are laid down in sections \ref{bloch} and \ref{purity} respectively. Marginals of three party systems and the issues pertaining to $ k- $ copy nonlocality and teleportation are investigated in sec. \ref{appl}. Finally, in sec. \ref{con}, we note our concluding remarks. 
	\section{Preliminaries and Notations}\label{pre}
	In this section, we introduce some preliminary concepts and some notations that are required for this article. Here $ \mathsf{H_{AB}} $ denotes the finite dimensional Hilbert space of the composite system $ \mathsf{AB} $. $ \mathfrak{B}(\mathsf{H}) $ denotes the bounded linear operators acting on the Hilbert space $ \mathsf{H} $.\\
	
	\textbf{Fully Entangled Fraction:} For a density matrix $ \rho_{d \otimes d} \in \mathfrak{B}(\mathsf{H_{AB}}) $, the fully entangled fraction(FEF) is given as \cite{horodecki}, 
	\begin{equation}
		F(\rho_{d \otimes d})=\Max\limits_{U_l}\:\langle \psi_{+}|(I\otimes U_l^{\dagger})\: \rho_{d \otimes d} \:(I\otimes U_l)|\psi_{+}\:\rangle
	\end{equation}
	the maximization being done over all local unitary operators $ U_l $ and $|\psi_{+}\rangle =\frac{1}{\sqrt{d}}\:\sum_{i=0}^{d-1}\:|ii \rangle $ is the maximally entangled state.
	\par $ \mathfrak{F} $ denotes the set of density matrices whose FEF is bounded by $ \frac{1}{d} $, i.e., $ \mathfrak{F} = \lbrace \rho_{d \otimes d} \in \mathfrak{B}(\mathsf{H_{AB}}): F(\rho_{d \otimes d}) \le \frac{1}{d} \rbrace $. If $ U_{nl} $ denotes a non-local unitary operator, then the states having absolute FEF is denoted by $ \mathfrak{AF} = \lbrace \sigma_{af} \in \mathfrak{B}(\mathsf{H_{AB}}): F(U_{nl} \sigma_{af} U_{nl}^\dagger) \le \frac{1}{d}, \forall U_{nl} \rbrace $. One of the objectives of the present paper is to characterize $ \mathfrak{AF} $ and address the membership problem in the class.\\
	
	\textbf{Absolute Separable States ($ \mathfrak{AS} $) :}
	Unlike a single qubit system,  two or more qubits can be entangled or separable in nature. It depends on the basis of the underlying composite Hilbert space \cite{Hahn}. A state in one basis can be entangled, while it can be separable in another basis. Here, it is interesting to note that there are states which are separable in all basis. Such states are known as Absolute Separable states ($ \mathfrak{AS} $)  \cite{purity2}. In other words we can interpret absolutely separable states $ \sigma_{as} $ as those states for which $ U \sigma_{as} U^\dagger $ is separable for all possible global unitary operations $ U $.
	
	\textbf{Absolute Bell- CHSH local states: } A state is considered to be Bell-CHSH local if it does not violate the Bell-CHSH inequality. A necessary and sufficient condition to satisfy the Bell-CHSH inequality for two qubits is $ M(\rho_l) \le 1 $, where $ M(\rho_l) $ is the sum of the largest two eigenvalues of $ T^\dagger T $, $ T $ being the correlation matrix for the two qubit density matrix $ \rho_l $ \cite{bell}. A state $ \sigma_{al} $ is said to be absolutely Bell-CHSH local if it does not violate the Bell-CHSH inequality even after non-local unitary action \cite{Bell-CHSH}. Precisely, the absolute Bell-CHSH local set is denoted by $ \mathfrak{AL} = \lbrace \sigma_{al} : M(U_{nl} \sigma_{al} U_{nl}^\dagger) \le 1 \rbrace $.\\
	
	\textbf{Absolute Conditional Von Neumann Entropy Non Negative States: }The von Neumann entropy of a quantum state $ \rho_{ab} $ is denoted by $ S(\rho_{ab}) = -tr(\rho_{ab} \log_2 \rho_{ab}) $, with the conditional entropy as $ C(\rho_{ab}) = S(\rho_{ab}) - S(\rho_{b}) $. The set $ \mathfrak{AC} $ represents the two qubit states whose conditional entropy remains non-negative even under global unitary action. A characterization of $ \mathfrak{AC} $ was done in \cite{spatro}. 
	\section{Characterization of $ \mathfrak{AF} $}\label{charac}
	In this section we find a restriction on the spectrum of a density matrix to decide it's membership in $ \mathfrak{AF} $.\\ 
	For density matrices in $ d \otimes d $ dimensions, a state belongs to $ \mathfrak{AF} $, if $ \langle v_{max} | U \rho_{d \otimes d} U^\dagger | v_{max} \rangle \le \frac{1}{d} $ for any $ U $, where $ | v_{max} \rangle  $ is a maximally entangled state. This is equivalent to $  \langle w_{pure} |  \rho_{d \otimes d}  | w_{pure} \rangle \le \frac{1}{d} $ for all pure states $ |w_{pure} \rangle  $. This further entails that all the eigenvalues of $ \rho_{d \otimes d} $ is $ \le \frac{1}{d} $. Hence we have the following theorem, 
	\begin{theorem}
		A two qudit state $ \rho_{d \otimes d} \in \mathfrak{AF} $ iff  $ \lambda_{max} \le \frac{1}{d} $, where $ \lambda_{max} $ is the maximum eigenvalue of $ \rho_{d \otimes d} $.
	\end{theorem}
	
	\par In what follows below, we give a functional analytic characterization of $ \mathfrak{AF} $. This will allow for the construction of witness operators that can detect the states whose FEF can be increased beyond $\frac{1}{d}$ after the application of global unitary operator. These are potentially \textit{useful} states for various quantum information processing protocols like teleportation.
	
	\par The geometric form of the Hahn-Banach theorem enables one to construct hyper-planes separating a set from a point outside it. However, its application requires that the set must be convex and compact, which we prove below for the set of our concern $ \mathfrak{AF} $.	
	\begin{theorem}
		$ \mathfrak{AF} $ is convex and compact
	\end{theorem}
	
	\begin{proof}
		(i) $ \mathfrak{AF} $ is convex: \\
		
		Let $\rho_{1},\rho_{2}\in$ $ \mathfrak{AF} $. 
		Let $\rho_{f}=\lambda \rho_{1} + (1-\lambda)\rho_{2}$ ; $\lambda\in[0,1]$. Consider an arbitrary unitary operator $U$. Then,
		\begin{align}
			U \rho_{f} U^{\dagger}= U(\lambda \rho_{1} + (1-\lambda)\rho_{2}) U^{\dagger}
		\end{align}
		\begin{equation} 
			\begin{split}
				F(U \rho_{f} U^{\dagger}) & =  F(U(\lambda \rho_{1} + (1-\lambda)\rho_{2}) U^{\dagger}) \\
				& = \lambda F(U\rho_{1} U^{\dagger})+(1-\lambda)F(U\rho_{2} U^{\dagger})
				\le \frac{1}{d}
			\end{split}
		\end{equation}
		Hence, $ \mathfrak{AF} $ is convex.\\
		
		(ii) $ \mathfrak{AF} $ is compact:\\
		Here, one can retrace the steps done in \cite{Ganguly14}, to prove that an arbitrary limit point of $ \mathfrak{AF} $ is convex will belong to the set itself, proving it to be closed. Compactness follows from the fact that $ \mathfrak{AF} $ is a subset of the compact set $ \mathfrak{F} $ \cite{closed}, whose FEF is bounded above by $ \frac{1}{d} $.
	\end{proof}
	
	\par The theorem above now entails that there exists witness operators $ \mathbf{S} $ with the following two properties, \\
	(i)  $ Tr[\mathbf{S} \sigma] \ge 0, \forall \sigma \in \mathfrak{AF}  $  \\
	(ii) $\exists \rho \in \mathfrak{F}-\mathfrak{AF},  Tr[\mathbf{S} \rho] < 0 $ \\
	Here in this work, we also give a prescription to construct the witness. One such example is given below: \\
	Consider $ \rho \in \mathfrak{F}-\mathfrak{AF} $. Therefore, there exists a unitary operator $ U $, such that $ F(U \rho U^\dagger) > \frac{1}{d} $. Hence, there exists a teleportation witness operator $ W $ that detects $ U \rho U^\dagger $ \cite{Ganguly11}. Now, consider $ \mathbf{S} = U^\dagger W U $. 
	It is now easy to see that $ Tr[\mathbf{S} \rho ]= Tr[U^\dagger W U \rho]= Tr[W U \rho U^\dagger ] <0 $.\\
	If $ \sigma \in \mathfrak{AF} $, then $ Tr[\mathbf{S} \sigma]= Tr[U^\dagger W U \sigma] = Tr[ W U \sigma U^\dagger] \ge 0$. The last inequality follows from the fact that $ W $ being a teleportation witness will give a non-negative expectation value over all states belonging to $ \mathfrak{F} $.
	\section{Illustrations and Decomposition of the Witness Operator}\label{decom}
	Here, in this section, we provide illustrations from different dimensions on the existence of states with absolute FEF and also the application of the witness operators constructed above to detect states whose FEF can be enhanced beyond $ \frac{1}{d} $, by global unitary operations. 
    \par \textit{States in two qubits-}
	As a first example, consider the following two-qubit state $X_{1}$,
	\begin{align}
		X_{1}=\frac{2}{9}\:|\phi_{2}^{+}\rangle \langle \phi_{2}^{+}|+\frac{1}{9}\:|01\rangle \langle 01|+\frac{1}{9}\:|10\rangle \langle 10|+\frac{5}{9}\:|00\rangle \langle 00|,
	\end{align}
	where $|\phi_{2}^{+}\rangle =\frac{1}{\sqrt{2}}\:\left (|00 \rangle +|11 \rangle \right )$. As the FEF of the density matrix $X_{1}$ is $\frac{1}{2}$ \cite{upper bound}, it is not useful for teleportation under the standard protocol \cite{wernerstate, horodecki}.\\
	However, the application of the unitary 
	\begin{equation}
		U_{1}=\begin{pmatrix}
			\frac{1}{\sqrt{2}} & 0 &  0& -\frac{1}{\sqrt{2}}\\ 
			0 &1  &  0&0 \\ 
			0 & 0 &  1&0\\ 
			\frac{1}{\sqrt{2}} &0  & 0 &\frac{1}{\sqrt{2}} 
		\end{pmatrix}
	\end{equation}
	results in the modified state $X_{1}^{'}= U_{1}\:X_{1}\:U_{1}^{\dagger}$. The density matrix representation of the state is given by,
	\begin{align*}
		X_{1}^{'}=\frac{5}{9}\:|\phi_{2}^{+}\rangle \langle \phi_{2}^{+}|+\frac{1}{9}\:|01 \rangle \langle 01|+\frac{1}{9}|10\rangle \langle 10|+\frac{2}{9}\:|11\rangle\langle 11|,
	\end{align*}
	and the FEF is $\frac{2}{3}$.
	The witness for the original state is given by, 
	\begin{align}
		\mathbf{S_{1}}=\frac{1}{2}\:\left (|01 \rangle \langle 01|+|10\rangle \langle 10|+|11 \rangle \langle 11|-|00\rangle \langle 00|\right ) \label{sepwit}
	\end{align}
	Now, for this witness we have $Tr(\mathbf{S_{1}}\:X_{1})=-\frac{1}{6} < 0$, thereby detecting that $ X_{1} \in \mathfrak{F}-\mathfrak{AF} $. It further indicates that the state $ X_{1} $ can be made useful for teleportation through non-local unitary action.\\ 
	
	Next, consider the state $ X_{2}$ \cite{Activating}
	\begin{align}
		X_{2}=q\:|\phi_{2}^{+}\rangle \langle \phi_{2}^{+}|+(1-q)\:|0\rangle \langle 0|\otimes |1\rangle \langle 1|,
	\end{align}
	where $|\phi_{2}^{+}\rangle =\frac{1}{\sqrt{2}}\:\left (|00 \rangle +|11 \rangle \right )$ and  $q\in (0,1]$. The FEF of $\ X_{2}$ is  greater than $\frac{1}{2}$ for $q\in (\frac{1}{2},1]$ \cite{Activating} and for $q\in (0,\frac{1}{2}] $, it's FEF is $\leq \frac{1}{2}$.
	Let us consider a global unitary operator,
	\begin{align}
		U_{2}=\frac{1}{2}\:\begin{pmatrix}
			-1 & \sqrt{2} &  0& -1\\ 
			0 &0  &  2&0 \\ 
			-\sqrt{2} & 0 &  0&\sqrt{2}\\ 
			1 &\sqrt{2}  & 0 &1
		\end{pmatrix}
	\end{align}
	Now, after the application of the global unitary operator $U_2$ on the state $X_2$ we have the transformed state as, 
	\begin{eqnarray}
		X_{2}^{'}=&&\frac{1}{2}\:\left (|00\rangle \langle 00|+|11\rangle \langle 11|\right)+{}\nonumber\\&& \left ( \frac{1}{2}-q\right )\:\left (|00\rangle \langle 11|+|11\rangle \langle 00|\right ),
	\end{eqnarray}
	where 	$X_{2}^{'}= U_{2}\:X_{2}\:U_{2}^{\dagger}$. FEF of $X_{2}^{'}$ is $\frac{1}{2}\:\left (1+\sqrt{(-1+2q)^2}\right )$. Hence, for $ 0< q < \frac{1}{2}$ and $\frac{1}{2}< q \leq 1$, its FEF  $ > \frac{1}{2} $. The witness for this state is, 
	\begin{align}
		\mathbf{S_{2}}=\frac{1}{2}\:\left (|00\rangle \langle 00|-|01 \rangle \langle 01|+|10\rangle \langle 10|+|11 \rangle \langle 11|\right ), \label{entwit}
	\end{align}
	which gives, $Tr(\mathbf{S_{2}} \:X_{2})=q-\frac{1}{2} < 0$,  when $q<\frac{1}{2}$. 
	
	\par \textit{State in two qutrits-}	Next, in this subsection we consider from the two qutrits system. One such example is the state $Y_{3}$ which is given by, \cite{Activating}, 
	\begin{align}\label{kcopy}
		Y_{3}=q\:|\phi_{3}^{+}\rangle \langle \phi_{3}^{+}|+(1-q)\:|0\rangle \langle 0|\otimes |1\rangle \langle 1|,
	\end{align}
	where $|\phi_{3}^{+}\rangle =\frac{1}{\sqrt{3}}\:\left (|00 \rangle +|11 \rangle+|22 \rangle  \right )$ and  $q\in (0,1]$. The FEF of $\ Y_{3}$ is  greater than $\frac{1}{3}$ for $q\in (\frac{1}{3},1]$ \cite{Activating} and for $q\in (0,\frac{1}{3}] $, its FEF is $\leq \frac{1}{3}$.
	Let us consider a global unitary operator $U_3$ given by,
	\begin{align}
		U_{3}=\frac{1}{2}\:\begin{pmatrix}
			-1 & \sqrt{2} & 0 & 0 & 0 & 0 & 0 & 0 & -1 \\ 
			0&0  &2  &0  &0 &0  & 0 &0  &0 \\ 
			0&0  &0  &2  &0  &0  &0  &0  &0 \\ 
			0&0  &0  &0  &2  &0  &0  &0  &0 \\ 
			0& 0 &0  &0  &0  &2  &0  &0  &0 \\ 
			0& 0 &0  &0  &0  &0  &2  & 0 &0 \\ 
			0& 0 &0  &0  &0  &0  & 0 &2  &0 \\ 
			-\sqrt{2}&0  &0  &0  &0  &0  &0  &0  &\sqrt{2} \\ 
			1& \sqrt{2} & 0 &0  &0  &0  &0  &0  & 1
		\end{pmatrix}
	\end{align}
	
	The state $Y_{3}^{'}$ after the application of the global unitary operator $U_3$ is given by,

	\begin{multline*}
			Y_{3}^{'}= U_{3}\:Y_{3}\:U_{3}^{\dagger}\\
		 =\frac{3-q}{6}\:\left (|00\rangle \langle 00|+|22\rangle \langle 22|\right)+\frac{3-5q}{6}\:\left (|00\rangle \langle 22|+|22\rangle \langle 00|\right)+\\\frac{q}{3}\:\left (|10\rangle \langle 10|+|10 \rangle \langle 22|+|22\rangle \langle 10|\right) -
		\frac{q}{3}\:\left (|00\rangle \langle 10|+|10 \rangle \langle 00|\right)
	\end{multline*}

	For the state $ Y_3 $ a suitable witness operator is $\mathbf{S_3}=U_{3}^{\dagger} W U_{3}$ as it gives $Tr(\mathbf{S_{3}}.Y_{3})=\frac{-1+2q}{3}<0$ for $q\in (0,\frac{1}{3}]$. Here, $ W $ represents the teleportation witness operator in two qutrits given in \cite{Ganguly11}. 
	
	\par \textit{Isotropic state in two qudits-} The general expression for a bipartite isotropic state in $d\times d$ dimensions is,  
	\begin{align*}
		X_{iso}=\beta|\psi^{+}\rangle \langle\psi^{+}|+\frac{1-\beta}{d^2}I, 
	\end{align*}
	where $|\psi^{+}\rangle =\frac{1}{\sqrt{d}}\sum_{i=0}^{d-1}|ii\rangle $ with $-\frac{1}{d^2-1}\leq\beta\leq 1$ \\
	
	The eigen values of $X_{iso}$ are $\frac{\beta(d^2-1)+1}{d^2}$ and $\frac{1-\beta}{d^2}$. This state is separable for $-\frac{1}{d^2-1}\leq\beta\leq \frac{1}{d+1}$ and entangled for $\frac{1}{d+1}<\beta\leq 1$. All entangled isotropic states are useful for teleportation. For the given separable range, the maximum eigen values of $X_{iso}$ is $\frac{1}{d}$, which leads to the conclusion that for the separable range of the parameter $\beta$, isotropic state belong to $ \mathfrak{AF} $.
	\par \textit{Relation between different absolute classes-} 
	It is trivial to see that $ \mathfrak{AS} $ is a subset of $ \mathfrak{AF} $ as the FEF of a separable state lies in the range $\frac{1}{d^2}\leq F(\rho)\leq \frac{1}{d}$ \cite{properties}. That is a proper subset can be seen as 
	\begin{equation*}
		\rho = \frac{5}{10}|00\rangle \langle 00|+\frac{3}{10}|01\rangle \langle 01|+\frac{2}{10}|10\rangle \langle 10|  
	\end{equation*}
	This belongs to $ \mathfrak{AF} $ but not to $ \mathfrak{AS} $. \\
	
	Although we cannot make any general comment on the relation between $ \mathfrak{AF} $ and $ \mathfrak{AL} $, we observe that the separable two qubit isotropic state belongs to both $ \mathfrak{AF} $ and $ \mathfrak{AL}$. In a similar line, the separable two qubit Werner state\cite{EPR} belongs to both $ \mathfrak{AF} $ and $ \mathfrak{AC}$. 
	
	\par \textit{Decomposition of the witness operator-}
	The decomposition of the witness in terms of local observables is necessary for practical implementation in a laboratory \cite{pauli}. In this section we show the decomposition of two witnesses in terms of both Pauli  basis and polarization states. In the two qutrit scenario,we give the decomposition in terms of the Gell-Mann matrices \cite{gellmann}. 
	
	When we consider the witness $\mathbf{S_{1}}$ given in equation (\ref{sepwit}), we get the decomposition in terms of Pauli matrices as
	\begin{align*}
		\mathbf{S_{1}} &=\frac{1}{4}\:(I\otimes I-\sigma_{z}\otimes I- I\otimes \sigma_{z}-\sigma_{z}\otimes \sigma_{z})
	\end{align*}
	Similarly, $\mathbf{S_{2}}$ (\ref{entwit}), in terms of Pauli matrices is given by
	\begin{align*}
		\mathbf{S_{2}} &=\frac{1}{4}\:(I\otimes I-\sigma_{z}\otimes I+ I\otimes \sigma_{z}+\sigma_{z}\otimes \sigma_{z})
	\end{align*}
	
	One may also consider $|H\rangle=|0\rangle$, $|V\rangle=|1\rangle$, $|D\rangle=\frac{|H \rangle +|V \rangle}{\sqrt{2}}$, $|F\rangle=\frac{|H \rangle -|V \rangle}{\sqrt{2}}$ , $|L\rangle=\frac{|H \rangle +i\:|V \rangle}{\sqrt{2}}$, $|R\rangle=\frac{|H \rangle -i\:|V \rangle}{\sqrt{2}}$ as the horizontal, vertical, diagonal, and the left and right circular polarization states respectively. Using the basis we can  write the witness $\mathbf{S_{1}}$ as follows:
	\begin{align*}
		\mathbf{S_{1}} &=\frac{1}{2}\:(-|HH\rangle \langle HH|+|HV\rangle \langle HV|\\ &+|VH\rangle \langle VH|+|VV\rangle \langle VV|)
	\end{align*}
	
	Similarly, for the witness $\mathbf{S_{2}}$, one can write as follows:
	\begin{align*}
		\mathbf{S_{2}} &=\frac{1}{2}\:(|HH\rangle \langle HH|-|HV\rangle \langle HV|\\ & +|VH\rangle \langle VH|+|VV\rangle \langle VV|)
	\end{align*}
	In the two qutrit case,the witness operator $\mathbf{S_{3}}$ can also be decomposed in terms of the Gell-Mann matrices, which are 
	\[
	I=\begin{bmatrix}
		1  &  0 & 0      \\
		0  &  1 & 0 \\
		0  &  0 &  1  
	\end{bmatrix}
	, 
	\Lambda_{1}=\begin{bmatrix}
		0  &  1 & 0      \\
		1  &  0 & 0 \\
		0  &  0 &  0  
	\end{bmatrix} 
	,
	\Lambda_{2}=\begin{bmatrix}
		0  &  -i & 0      \\
		i  &  0 & 0 \\
		0  &  0 &  0  
	\end{bmatrix} ,
	\]
	\[
	\Lambda_{3}=\begin{bmatrix}
		1  &  0 & 0      \\
		0  &  -1 & 0 \\
		0  &  0 &  0  
	\end{bmatrix}
	, 
	\Lambda_{4}=\begin{bmatrix}
		0  &  0 & 1      \\
		0  &  0 & 0 \\
		1  &  0 &  0  
	\end{bmatrix} 
	,
	\Lambda_{5}=\begin{bmatrix}
		0  &  0 & -i      \\
		0  &  0 & 0 \\
		i  &  0 &  0  
	\end{bmatrix} ,
	\]
	\[
	\Lambda_{6}=\begin{bmatrix}
		0  &  0 & 0      \\
		0  &  0 & 1 \\
		0  &  1 &  0  
	\end{bmatrix}
	, 
	\Lambda_{7}=\begin{bmatrix}
		0  &  0 & 0      \\
		0  &  0 & -i \\
		0  &  i &  0  
	\end{bmatrix} 
	,
	\Lambda_{8}=\frac{1}{\sqrt{3}}\begin{bmatrix}
		1  &  0 & 0      \\
		0  &  1 & 0 \\
		0  &  0 &  -2  
	\end{bmatrix} 
	\]
The decomposition stands as:
 \begin{widetext}
	\begin{align*}
		S_{3}=\frac{1}{3}\:\left (I\otimes I -\frac{1}{\sqrt{2}}(\Lambda_{1}\otimes \Lambda_{6}-\Lambda_{2}\otimes \Lambda_{7})-2( A\otimes B)- B\otimes C \right ),
	\end{align*}
\end{widetext}
	where \begin{equation*}
		\begin{aligned}
			A & =\frac{1}{2} \Lambda_{3}+\frac{1}{2\sqrt{3}}\Lambda_{8}+\frac{1}{3}I \\
			B & =\frac{-1}{2} \Lambda_{3}+\frac{1}{2\sqrt{3}}\Lambda_{8}+\frac{1}{3}I \\
			C & =\frac{-1}{\sqrt{3}}\Lambda_{8}+\frac{1}{3}I
			\label{current_rel1}
		\end{aligned}
	\end{equation*}

	\section{Characterization in terms of Bloch parameters} \label{bloch}
	It is important to characterize the membership problem in $ \mathfrak{AF} $, in terms of Bloch parameters of the density matrix. In this context, we consider several classes below.
	
	Let $\rho$ be the general  density matrix in two qubits. The representation of $\rho$ in terms of Bloch parameters \cite{upper bound},
	\begin{widetext}
	\begin{align*}
		\rho=\frac{1}{4}\:I\otimes I+ \frac{1}{2}\sum_{i=1}^{3} a_i \sigma_i \otimes I + \frac{1}{2} \sum_{j=1}^{3} b_j I \otimes \sigma_j + \sum_{i,j=1}^{3} t_{ij} \sigma_i \otimes \sigma_j 
	\end{align*}
\end{widetext}
	Here $a_i=\frac{1}{2}\:Tr\left\{\rho\:\sigma_{i}\otimes I\right\}$, $b_j=\frac{1}{2}\:Tr\left\{\rho\:I\otimes \sigma_{j}\right\}$ and $t_{ij}=\frac{1}{4}\:Tr[\rho(\sigma_{i}\otimes \sigma_{j})]$, and $\sigma_{i}$
	are the Pauli matrices.
\par \textit{Class-I -} Here, we consider the  class of states in the Bloch sphere for which we have $\vec{a}=\vec{b}=0$, $T=[t_{ii}];(i=1,2,3)$.\\
	Let the density matrix of this class be given by $\rho_{ab}$. Eigen values of $\rho_{ab}$ are
		\begin{align*}
		&\left\{\frac{1}{4}\:(1-4\:(t_{11}+t_{22}+t_{33})),\frac{1}{4}\:(1+4\:(t_{11}+t_{22}-t_{33})),\right.\\
		&\left.\frac{1}{4}\:(1+4\:(t_{11}-t_{22}+t_{33})),\frac{1}{4}\:(1+4\:(t_{22}+t_{33}-t_{11})) \right\}
	\end{align*}
	The state $\rho_{ab}\in \mathfrak{AF}$ if  we have,
\begin{align*}
	&\text{Max eigen value} \leq \frac{1}{2}\\
	\implies& Max \{ \left[1-4(t_{11}+t_{22}+t_{33}),1+4(t_{11}+t_{22}-t_{33}),\right .\\
	&\left.1+4(t_{11}-t_{22}+t_{33}),1+4(t_{22}+t_{33}-t_{11}) \right] \} \leq 2\\
	\implies& 	Max \{ \left[-(t_{11}+t_{22}+t_{33}),(t_{11}+t_{22}-t_{33}),\right.\\
	&\left.(t_{11}-t_{22}+t_{33}),(t_{22}+t_{33}-t_{11}) \right] \} \leq \frac{1}{4}
\end{align*}
	
	\par \textit{Class-II-} Here, we consider the class of states which are diagonal in the computational basis. The density matrix representation of such a class of state is given by, 
	\begin{align*}
		\rho_{comp}=a|00\rangle \langle 00|+b|01\rangle \langle 01|+c|10\rangle \langle 01|+d|11\rangle \langle 11|,
	\end{align*}
	where $a\geq b\geq c\geq d$.
	Here, $t_{11}=t_{22}=0$ and $t_{33}=\frac{a-b-c+d}{4}$,
	$\vec{a}=(0,0,a_{3})$ with $a_{3}=\frac{a+b-c-d}{2}$ and $\vec{b}=(0,0,b_{3})$ with $b_{3}=\frac{a+c-b-d}{2}$.\\
	
	$\rho_{comp}\in \mathfrak{AF}$ if
	\begin{align*}
		&Max \{ (|a_{3}\mp b_{3}|\mp 2t_{33}) \} \leq \frac{1}{2} \\
		\implies& Max \{ (|a-d|+\frac{a-b-c+d}{2},\\ &
		|b-c|+\frac{a-b-c+d}{2}) \} \leq \frac{1}{2}
	\end{align*}
	
	\textbf{Case-1: when $a>d$,  $b>c$}\\
	
	$\rho_{comp}\in \mathfrak{AF}$ if
	\begin{align*}
		& Max \{ (|a-d|+\frac{a-b-c+d}{2},|b-c|+\frac{a-b-c+d}{2}) \} \leq \frac{1}{2}\\
		&\implies Max \{ (4a-1, 1-4c) \} \leq \frac{1}{2}
	\end{align*}
	
	\textbf{Case-2: when $a>d$, $b<c$}\\
	
	$\rho_{comp}\in \mathfrak{AF}$ if
	\begin{align*}
		& Max \{ (|a-d|+\frac{a-b-c+d}{2},|b-c|+\frac{a-b-c+d}{2}) \} \leq \frac{1}{2}\\
		&\implies Max \{ (4a-1, 1-4b) \} \leq \frac{1}{2}
	\end{align*}
	
	\textbf{Case-3: when  $a<d$,  $b<c$}\\
	
	$\rho_{comp}\in \mathfrak{AF}$ if
	\begin{align*}
		& Max \{ (|a-d|+\frac{a-b-c+d}{2},|b-c|+\frac{a-b-c+d}{2}) \} \leq \frac{1}{2}\\
		&\implies Max \{ (4d-1, 1-4b) \} \leq \frac{1}{2}
	\end{align*}
	
	\textbf{Case-4: when $a<d$, $b>c$}\\
	
	$\rho_{comp}\in \mathfrak{AF}$ if
	\begin{align*}
		& Max \{ (|a-d|+\frac{a-b-c+d}{2},|b-c|+\frac{a-b-c+d}{2}) \} \leq \frac{1}{2}\\
		&\implies Max \{ (4d-1, 1-4c) \} \leq \frac{1}{2}
	\end{align*} 
	
	\section{Bounds using purity}\label{purity}
	In the two qubit case,we use optimization procedure to answer the two questions (i) What is the maximum purity above which there cannot be any state which is in $ \mathfrak{AF} $? (ii) What is the minimum purity, below which all states are in $ \mathfrak{AF} $?
	
	\par \textit{Maximum purity-} Let us assume that $\lambda_{1}$, $\lambda_{2}$, $\lambda_{3}$ and $\lambda_{4}$ be the four eigen values of a density matrix which are arranged in descending order i.e., $\lambda_{1}\geq \lambda_{2}\geq\lambda_{3}\geq \lambda_{4}$. The condition of a state to be lie in $ \mathfrak{AF} $ is the maximum eigen value of the density matrix should be less than or equal to $\frac{1}{2}$. i.e., $\lambda_{Max}\leq \frac{1}{2}$. Now, the problem of maximum purity above which there can't be any $ \mathfrak{AF} $ class is framed as:
	
	\begin{equation*}
		\begin{aligned}
			& {\text{Maximize}}
			& &\lambda_{1}^{2}+\lambda_{2}^{2}+\lambda_{3}^{2}+\lambda_{4}^{2},\\
			& \text{subject to:} & (i)&~  \lambda_{1}\leq \frac{1}{2} \\
			& & (ii)& ~ \lambda_{1}+\lambda_{2}+\lambda_{3}+\lambda_{4}=1,\\ 
			& &(iii) &~ 1\geq\lambda_{1}\geq \lambda_{2}\geq  \lambda_{3}\geq \lambda_{4}\geq 0  \\
		\end{aligned}
	\end{equation*}
	
	A numerical computation yields the threshold value as $\frac{1}{2}$. i.e., states whose purity exceeds $\frac{1}{2}$, cannot belong to $ \mathfrak{AF} $.\\
	
	\par \textit{Minimum purity-} Now, consider the second question, 
	\begin{equation*}
		\begin{aligned}
			& {\text{Minimize}}
			& &\lambda_{1}^{2}+\lambda_{2}^{2}+\lambda_{3}^{2}+\lambda_{4}^{2},\\
			& \text{subject to:} &(i) &  \lambda_{1}> \frac{1}{2} \\
			& &(ii) & ~ \lambda_{1}+\lambda_{2}+\lambda_{3}+\lambda_{4}=1,\\ 
			& & (iii)&~ 1\geq\lambda_{1}\geq \lambda_{2}\geq  \lambda_{3}\geq \lambda_{4}\geq 0  \\
		\end{aligned}
	\end{equation*}
	
	In this case, after numerical computations, we have the threshold value as $\frac{1}{3}$, a purity value below which all states will belong to $ \mathfrak{AF}. $ 
	\section{Applications}\label{appl}
	
	\subsection{k-copy Nonlocality and Teleportation}
	
	Nonlocality is a distinctive feature of quantum theory besides entanglement. We say that a state is Bell nonlocal if it violates a Bell's inequality. However, sometimes a single copy of a quantum state might not violate a Bell's inequality. Whereas, if we take several copies of the state, then it violates the inequality. Precisely, we say a state $ \rho $ is $ k- $ copy nonlocal if $ \rho^{\otimes k} $ is nonlocal for some $ k $ \cite{kcopy}. 
	\par An interesting link between teleportation and $ k- $ copy nonlocality was provided in \cite{kcopy}. The link was established through the fully entangled fraction. FEF $ > 1/d $ constitutes a necessary and sufficient condition for quantum advantage in teleportation. It was shown in \cite{kcopy}, that the condition FEF $ > 1/d $ also provides a sufficient condition for a state to be $ k- $ copy nonlocal. 
	
	\par Consider now, the following two cases: \\
	(i) Consider that a density matrix $ \in \mathfrak{F}-\mathfrak{AF} $. Therefore, initially it is not useful for teleportation. However, it is found that after the application of a suitable unitary operation the transformed state has FEF $ > 1/d $. Thus the transformed state becomes useful for teleportation and also can be said to be $ k- $ copy nonlocal following \cite{kcopy}. The witness operators constructed here can be used to identify such states. \\
	(ii) If a density matrix belongs to $ \mathfrak{AF} $, which is identified through our criteria on eigenvalues, then such a state cannot be made useful for teleportation even with global unitary action. However, in this case we will not be able to comment on its $ k- $ copy nonlocality.

	As an example if we take the state mentioned in Eq. (\ref{kcopy}), then initially the state is not useful for teleportation as its FEF $ \le 1/3 $.However after application of an appropriate unitary gate it can be made useful. The witness $\mathbf{S_3}$ detects this phenomenon. Thus it can also be definitely said that the transformed state is $ k- $ copy nonlocal. However if one takes a state from $ \mathfrak{AF} $, then even with any unitary operation the state will remain useless pertaining to teleportation. 
	
	\subsection{Marginals of tripartite states}
	We study the reduced subsystems of a three qubit systems in the context of their membership in $ \mathfrak{AF} $. 
	\par \textit{Pure three qubit states-} A pure tripartite state is given by \cite{Acin}, 
	\begin{equation}\label{pure}
		|\psi_{3}\rangle=x_0 |000\rangle+x_1 e^{\imath \theta}|100\rangle+x_2 |101\rangle+x_3 |110\rangle\rangle+x_4|111\rangle
	\end{equation}
	where $x_i$$\geq$$0,$ $0$$\leq$$\theta$$\leq$$\pi$ and $\sum_{i=0}^4 x_i^2$$=$$1.$\\
	Removing first qubit, the reduced state is given by:
	\begin{equation}\label{red1}
		\left(\begin{array}{cccc}
			x_0^2+x_1^2&e^{\imath \theta}x_1\lambda_2&e^{\imath \theta}x_1x_3&e^{\imath \theta}x_1x_4\\
			e^{-\imath \theta}x_1x_2&x_2^2&x_2x_3&x_2x_4\\
			e^{-\imath \theta}x_1x_3&x_2x_3&x_3^2&x_3x_4\\
			e^{-\imath \theta}x_1x_4&x_2x_4&x_3x_4&x_4^2\\
		\end{array} \right)
	\end{equation}
	Eigenvalues of the reduced state(Eq.(\ref{red1})) are,  $0,0,\frac{1}{2}(1\pm\sqrt{S_1}),$  where $S_1$ is given by $1+4(x_2^2+x_3^2-x_4^2+x_1^2x_4^2+x_3^2(-1+x_1^2+2x_4^2)
	+x_2^2(-1+x_1^2+2x_3^2+2x_4^2))$. 
	
	\par Consequently, the reduced state $ \in \mathfrak{AF} $ if:
	\begin{equation}
		-x_4^2+x_1^2x_4^2+x_3^2(x_1^2+2x_4^2)
		+x_2^2(x_1^2+2x_3^2+2x_4^2)=-\frac{1}{4}
	\end{equation}
	Removing second qubit, the reduced state is given by:
	\begin{equation}\label{red2}
		\left(\begin{array}{cccc}
			x_0^2&0&e^{-\imath \theta}x_0 x_1&x_0x_2\\
			0&0&0&0\\
			e^{\imath \theta}x_0x_1&0&x_1^2+x_3^2&e^{\imath \theta}x_1x_2+x_3x_4\\
			x_0x_2&0&e^{-\imath \theta}x_1x_2+x_3x_4&x_2^2+x_4^2\\
		\end{array} \right)
	\end{equation}
	The eigenvalues are $0,0,\frac{1}{2}(1\pm\sqrt{S_2}),$  where $S_2$ is given by
	$ 1-4(x_0^2x_3^2+x_2^2x_3^2-2\cos(\theta)x_1x_2x_3x_4+x_0^2x_4^2+x_1^2x_4^2)$. 
	\par Consequently, the reduced state $ \in \mathfrak{AF} $ if,
	\begin{equation}
		x_0^2x_3^2+x_2^2x_3^2-2\cos(\theta)x_1x_2x_3x_4+x_0^2x_4^2+x_1^2x_4^2=\frac{1}{4}.
	\end{equation}
	
	Removing third qubit, the reduced state is given by:
	\begin{equation}\label{red3}
		\left(\begin{array}{cccc}
			x_0^2&0&e^{-\imath \theta}x_0x_1&x_0x_3\\
			0&0&0&0\\
			e^{\imath \theta}x_0x_1&0&x_1^2+x_2^2&e^{\imath \theta}x_1x_3+x_2x_4\\
			x_0x_3&0&e^{-\imath \theta}x_1x_3+x_2x_4&x_3^2+x_4^2\\
		\end{array} \right)
	\end{equation}
The eigenvalues are $0,0,\frac{1}{2}(1\pm\sqrt{S_3}),$  where $S_3$ is given by
	$  1-4(x_0^2x_2^2+x_2^2x_3^2-2\cos(\theta)x_1x_2x_3x_4+x_0^2x_4^2+x_1^2x_4^2)$
	\par Consequently, the reduced state $ \in \mathfrak{AF} $ if,
	\begin{equation}
		x_0^2x_2^2+x_2^2x_3^2-2\cos(\theta)x_1x_2x_3x_4+x_0^2x_4^2+x_1^2x_4^2=\frac{1}{4}
	\end{equation}
	
	\par \textit{Mixture of GHZ and W states in three qubits-} Consider the state,
	\begin{eqnarray}\label{red4}
		\varrho&=&p |\phi_1\rangle\langle \phi_1| +(1-p) |\phi_2\rangle\langle \phi_2|,\,\textmd{where},\nonumber\\
		|\phi_1\rangle &=&\frac{1}{\sqrt{2}}(|000\rangle+|111\rangle)\nonumber\\
		|\phi_2\rangle &=&\frac{1}{\sqrt{3}}(|001\rangle+|010\rangle+|100\rangle)
	\end{eqnarray}
	Reduced state by eliminating anyone of the three qubits is given by:
	\begin{equation}\label{red5}
		\left(\begin{array}{cccc}
			\frac{2+p}{6}&0&0&0\\
			0&\frac{1-p}{3}&\frac{1-p}{3}&0\\
			0&\frac{1-p}{3}&\frac{1-p}{3}&0\\
			0&0&0&\frac{p}{2}\\
		\end{array} \right)
	\end{equation}
	Eigen values of the reduced matrix are,  $0,\frac{2}{3}(1-p),\frac{p}{2},\frac{2+p}{6}.$ Therefore, the reduced state $ \in \mathfrak{AF} $ for $p \in [0.25,1].$
	
	\textit{A state in three qutrits-} Consider now the state \cite{threequtrits}, 
	\begin{equation}\label{3qutrits}
		\rho_{3 \otimes 3 \otimes 3}=\alpha |GHZ_3 \rangle \langle GHZ_3|+\beta |\psi_3 \rangle \langle \psi_3|+\frac{1-\alpha-\beta}{27} I
	\end{equation}
    where, $ |GHZ_3\rangle=\frac{1}{\sqrt{3}}\sum_{i=0}^2 |iii\rangle $ and \\
    $ |\psi_3\rangle = \frac{|012\rangle+|021\rangle+|102\rangle+|120\rangle+|201\rangle+|210\rangle}{\sqrt{6}} $\\

	Each of the three reduced states is given by:
	
	\vskip 2cm
	
	\begin{widetext}
		
		\begin{equation}
			\left(\begin{array}{ccccccccc}
				
				\frac{1 + 2 \alpha -\beta}{9}&0&0&0&0&0&0&0&0\\
				
				0&\frac{2 - 2 \alpha +\beta}{18}&0&\frac{\beta}{6}&0&0&0&0&0\\
				
				0&0&\frac{2 - 2 \alpha +\beta}{18}&0&0&0&\frac{\beta}{6}&0&0\\
				
				0&\frac{\beta}{6}&0&\frac{2 - 2 \alpha +\beta}{18}&0&0&0&0&0\\
				
				0&0&0&0&\frac{1 + 2 \alpha -\beta}{9}&0&0&0&0\\
				
				0&0&0&0&0&\frac{2 - 2 \alpha +\beta}{18}&0&\frac{\beta}{6}&0\\
				
				0&0&\frac{\beta}{6}&0&0&0&\frac{2 - 2 \alpha +\beta}{18}&0&0\\
				
				0&0&0&0&0&\frac{\beta}{6}&0&\frac{2 - 2 \alpha +\beta}{18}&0\\
				
				0&0&0&0&0&0&0&0& \frac{1 + 2 \alpha -\beta}{9}\\
				
			\end{array} \right)
		\end{equation}
	\end{widetext}
	Eigen values of the reduced matrix are: $\{\frac{1}{9}(1 - \alpha -\beta),\frac{1}{9}(1 + 2 \alpha -\beta),\frac{1}{9}(1 - \alpha + 2\beta)\},$
	where each eigen value is of algebraic multiplicity $3.$ Reduced state will belong to $ \mathfrak{AF} $ if:

	\begin{equation}
		Max \{\frac{1}{9}(1 - \alpha -\beta),\frac{1}{9}(1 + 2 \alpha -\beta),\frac{1}{9}(1 - \alpha + 2\beta) \}\leq \frac{1}{3}
	\end{equation}

    Clearly, all the eigenvalues are less than $ 1/3 $, indicating that all the reduced states of the family(Eq.(\ref{3qutrits})) belong to $ \mathfrak{AF} $.

    The results above have a physical interpretation. Consider that three parties Alice, Bob and Charlie share a tripartite state. Now, if Alice and Bob collaborate i.e., they come together in a single lab and apply a global unitary on their combined system, then they can generate a state whose FEF becomes greater than the benchmark $ 1/d $ and thus the transformed state becomes useful for the task of teleportation. However, as our result shows, there will be instances in which even if two parties collaborate and apply a global unitary, the FEF cannot be enhanced beyond the benchmark.

	\section{Conclusions}\label{con}
	
	Certain manifestations of quantum states rely on the choice of basis of the underlying Hilbert space. One such important feature of a quantum state is it's fully entangled fraction, which measures its overlap with the maximally entangled state. A quantum state can have different FEF depending upon the choice of the basis. Basis transformations are brought about by unitary operations. In the current submission,we characterize the quantum states whose FEF cannot be increased beyond $ 1/d $, even with global unitary action. The benchmark $ 1/d $ is pertinent to several important information processing tasks that include teleportation. Through the prescription laid here, one can identify states whose FEF can be increased using appropriate unitary gates. However, if a state belongs to the absolute class as we show here, then even with global unitary operation, its FEF cannot be enhanced beyond the threshold. Illustrations from different dimensions highlight the necessary underpinnings.
	\par The work has been shown to carry significant implications if one observes the marginals of a three party system. Under certain restrictions the reduced subsystems have been shown to belong to the absolute class. This entails that even if two parties collaborate to apply a nonlocal unitary operator on their combined system, they cannot breach the threshold FEF value. Relevance of our work to $ k- $ copy nonlocality and teleportation is also discussed. 
	\par One useful direction of future research can be to study this behaviour of FEF in multipartite systems and also for bipartite $ d_1 \otimes d_2 $ systems where $ d_1 \neq d_2 $. 
	
	\section*{Acknowledgement}
	Tapaswini Patro would like to acknowledge the support from DST-Inspire fellowship No. DST/INSPIRE Fellowship/2019/IF190357. M.A.S. acknowledges the National Key  R \& D Program of China, Grant No. 2018YFA0306703.
	
	\section*{Author Contributions} 
	
	All authors contributed equally to the paper. 
	
	\section*{Data Availability Statement} 
	
	Data sharing is not applicable to this article as no datasets were generated or analysed during the current study.
	
	\bibliography{basename of .bib file}

\end{document}